\documentclass[letterpaper, 10 pt, conference]{ieeeconf} 
\IEEEoverridecommandlockouts                     
\usepackage{amsmath} % assumes amsmath package installed
\usepackage{amssymb}  % assumes amsmath package installed
\usepackage{graphicx}
\usepackage{epstopdf}
\usepackage{amsmath}
\usepackage{mathtools}
\usepackage{dsfont}
\let\labelindent\relax
\usepackage{enumitem}
\usepackage{tikz}
\usepackage{bm}
\usepackage{siunitx}
\usepackage{hyperref}
\DeclareMathOperator*{\argmin}{arg\,min}

\usepackage{balance}  
\usepackage{algorithm}
\usepackage{setspace}
\usepackage{caption}
\usepackage{subcaption}
\usepackage{varwidth}
\usepackage{cite}
\usepackage[noend]{algpseudocode}

\newcounter{algorithmctr}[section]
\renewcommand{\thealgorithmctr}{\thesection.\arabic{algorithmctr}}
{\refstepcounter{algorithmctr}\begin{list}{}{%
\setlength{\rightmargin}{0\linewidth}%
\setlength{\leftmargin}{.05\linewidth}
\setlength{\itemsep}{1pt}
\setlength{\parskip}{0pt}
\setlength{\parsep}{0pt}}%
\rmfamily\small
\item[]{\setlength{\parskip}{0ex}\hrulefill\par%
\nopagebreak{\bfseries\textsf{Algorithm \thealgorithmctr~}}}}%
{{\setlength{\parskip}{-1ex}\nopagebreak\par\hrulefill} \end{list}}
\IEEEoverridecommandlockouts
\newtheorem{assumption}{Assumption}
\newtheorem{theorem}{Theorem}

\title{\LARGE \bf
Task Decomposition for Iterative Learning Model Predictive Control
}

\author{Charlott Vallon and Francesco Borrelli% <-this % stops a space
\thanks{This research was sustained in part by fellowship support from the National Physical Science Consortium and the National Institute of Standards and Technology.}% <-this % stops a space
\thanks{Charlott Vallon and Francesco Borrelli are with the Department of Mechanical Engineering, University of California, Berkeley, Berkeley, CA 94701, USA
        {\tt\small \{charlott, fborrelli\} $@$ berkeley.edu}}%
}%

\begin{document}

\maketitle
\thispagestyle{empty}
\pagestyle{empty}

%%%%%%%%%%%%%%%%%%%%%%%%%%%%%%%%%%%%%%%%%%%%%%%%%%%%%%%%%%%%%%%%%%%%%%%%%%%%%%%%
\begin{abstract}
A task decomposition method for iterative learning model predictive control is presented. We consider a constrained nonlinear dynamical system and assume the availability of state-input pair datasets which solve a task $\mathcal{T}1$. 
Our objective is to find a feasible model predictive control policy for a second task, $\mathcal{T}2$, using stored data from $\mathcal{T}1$. Our approach applies to tasks $\mathcal{T}2$ which are composed of subtasks contained in $\mathcal{T}1$. 
In this paper we propose a formal definition of subtasks and the task decomposition problem, and provide proofs of feasibility and iteration cost improvement over simple initializations.
We demonstrate the effectiveness of the proposed method on autonomous racing and robotic manipulation experiments.
\end{abstract}

\section{Introduction}
%\subsection{Motivation}
Control design for systems repeatedly performing a single task has been studied extensively. Such problems arise frequently in practical applications \cite{AA, WangGaoDoyleILC2009} and examples range from autonomous cars racing around a track \cite{stanford, CarrauLinigerZhangLygerosECC2016, rosolia2017autonomousrace} to robotic system manipulators \cite{horowitz1993learning, robotilc1, van2010superhuman, wang2013adaptive}.
Iterative Learning Controllers (ILCs) aim to autonomously improve a system's closed-loop reference tracking performance at each iteration of a repeated task, while rejecting periodic disturbances \cite{AA, batchILC, newcite}.
In classical ILC, the controller uses tracking error data from previous task iterations to better track a provided reference trajectory during the current iteration.
Recent work has also explored reference-free ILC strategies for tasks whose goals are better defined in terms of an economic metric, rather than a reference trajectory. The controller again uses previous iteration data to improve closed-loop performance with respect to the chosen performance metric.
Examples include autonomous racing tasks (e.g. ``minimize lap time'') \cite{rosolia2017autonomousrace, melanie }, or optimizing flight paths for tethered energy-harvesting systems (e.g. ``maximize average power generation'') \cite{vermillion}. 

The aforementioned iterative learning methods require either a reference trajectory to track (classical ILC) or a feasible trajectory with which to initialize the iterative control algorithm (reference-free ILC). 
If the task changes, a new trajectory needs to be designed to match the new task. This can be challenging for complex tasks.

Many methods exist to use data collected from a task to efficiently solve variations of that task, including model-based and model-free methods. 
Here, we focus specifically on model-based methods for using stored trajectories from previous tasks in order to find feasible trajectories for new tasks. 
The authors in \cite{BG} propose running a desired planning method in parallel with a retrieve and repair algorithm that adapts reference trajectories from previous tasks to the constraints of a new task. Retrieve and repair was shown to decrease overall planning time, but requires checking for constraint violations at each point along a retrieved trajectory.
In \cite{6}, environment features are used to divide a task and create a library of local trajectories in relative state space frames. These trajectories are then pieced back together in real-time according to the features of the new task environment.
A trajectory library built using differential dynamic programming is used in \cite{liu2009standing} to design a controller for balance control in a humanoid robot. At each time step, a trajectory is selected from the library based on current task parameter estimates and a k-nearest neighbor selection scheme. A similar method is explored in \cite{7}, where differential dynamic programming is combined with receding horizon control. While these methods can decrease planning time, they verify or interpolate saved trajectories at every time step, which can be inefficient and unnecessary.

The authors in \cite{5} propose piecing together stored trajectories corresponding to discrete system dynamics only at states of dynamics transition. However, this method only applies to discontinuities in system dynamics, and does not generalize to other task variations.
In \cite{pereida2018data}, a static map is learned between a given reference trajectory and the input sequence required to make a linear time invariant system track that reference trajectory. Once learned, this map can be used to determine an input sequence that lets the linear system track a new reference trajectory.  

In this paper, our objective is to find a feasible trajectory to smartly initialize an Iterative Learning Model Predictive Controller (ILMPC) \cite{rosolia2016learning} for a new task, using data from previous tasks. 
ILMPC is a type of reference-free ILC that uses a \emph{safe set} to design a model predictive control (MPC) policy for an iterative control task. The ILMPC safe set is initialized using a feasible task trajectory.

We consider a constrained nonlinear dynamical system and assume the availability of a dataset containing states and inputs corresponding to multiple iterations of a task $\mathcal{T}1$.
This dataset can be stored explicitly (e.g. by human demonstrations~\cite{learnbydem} or an iterative controller \cite{rosolia2016learning}) or generated by roll-out of a given policy (e.g. a hand-tuned controller).
We introduce a Task Decomposition for ILMPC algorithm (TDMPC), and show how to use the stored $\mathcal{T}1$ dataset to efficiently construct a non-empty ILMPC safe set for task $\mathcal{T}2$ (a new variation of $\mathcal{T}1$), containing feasible trajectories for $\mathcal{T}2$. 

The contributions of this paper are twofold: 
\begin{enumerate}
    \item We first present how to build the aforementioned $\mathcal{T}2$ safe set using the TDMPC algorithm. TDMPC reduces the complexity to adapt trajectories from $\mathcal{T}1$ to a new task $\mathcal{T}2$ by decomposing task $\mathcal{T}1$ into different modes of operation, called subtasks. 
    The stored $\mathcal{T}1$ trajectories are adapted to $\mathcal{T}2$ only at points of subtask transition, by solving one-step controllability problems. 
    \item We prove that the resulting safe set based ILMPC policy is feasible for $\mathcal{T}2$, and the corresponding closed-loop trajectories have lower iteration cost compared to an ILMPC initialized using simple methods.
\end{enumerate}

\section{Problem Definition} \label{sec:PD} 
\subsection{Safe Set Based ILMPC}\label{ssec:ssilmpc}
We consider a system
\begin{align}\label{eq:VehicleModel}
    x_{k+1} &= f(x_k, u_k),
\end{align}
where $f(x_k, u_k)$ is the dynamical model, subject to the constraints
\begin{align}\label{eq:VehicleModelConstr}
    x_k \in \mathcal{X},~ u_k \in \mathcal{U}.
\end{align}
The vectors $x_k$ and $u_k$ collect the states and inputs at time step $k$. We pick a set $\mathcal{P} \subset \mathcal{X}$ to be the target set for an iterative task $\mathcal{T}$, performed repeatedly by system (\ref{eq:VehicleModel}) and defined by the tuple
\begin{align}\label{eq:task}
    \mathcal{T} = \left\{ \mathcal{X}, \mathcal{U}, \mathcal{P}\right\}.
\end{align}
\begin{assumption}\label{ass:controlInvariant}
$\mathcal{P}$ is a control invariant set \cite[Sec 10.9]{borrelliText}:
\begin{equation}
    \forall x_k \in \mathcal{P},~ \exists u_k \in \mathcal{U}~ : ~ x_{k+1} = f(x_k, u_k) \in \mathcal{P}. \nonumber
\end{equation}
\end{assumption}
Each task execution is referred to as an iteration.
The goal of an ILMPC is to solve at each iteration the optimal task completion problem:
    \begin{align}\label{eq:minTimeOC}
        V^\star_{0 \rightarrow T} (x_0) = & \min_{T, u_0, ..., u_{T-1}} \sum_{k=0}^T h(x_k, u_k) \\
        & ~~~~~~ \mathrm{s.t.}~~~~~ x_{k+1} = f(x_k, u_k), \nonumber \\ 
        & ~~~~~~~~~~~~~~~ x_k \in \mathcal{X}, u_k \in \mathcal{U} ~~ \forall k \geq 0,\nonumber \\
        & ~~~~~~~~~~~~~~~ x_{T} \in \mathcal{P}, \nonumber
    \end{align}
where $V^\star_{0 \rightarrow T} (x_0)$ is the optimal cost-to-go from the initial state $x_0$, and $h(x_k, u_k)$ is a chosen stage cost.

At the $j$-th successful task iteration, the vectors
\begin{subequations}\label{eq:iterationvecs}
    \begin{align}
E^j(\mathcal{T}) & = [{\bf{x}}^j, {\bf{u}}^j] \\ 
    {\bf{x}}^j & = [x_{0}^j,x_{1}^j,...,x^j_{T^j}], ~x^j_k \in \mathcal{X}~~ \forall k \in [0, {T^j}], \nonumber\\ & ~~~~~~~~~~~~~~~~~~~~~~~~~~~~~~~~~~ x_{T^j}^j \in \mathcal{P},\\
    {\bf{u}}^j & = [u_0^j,u_1^j,...u^j_{T^j}], ~u^j_k \in \mathcal{U}~~ \forall k \in [0, {T^j}],
    \end{align}
\end{subequations}
collect the inputs applied to system (1) and the corresponding state evolution. In (\ref{eq:iterationvecs}), $x_k^j$ and $u_k^j$ denote the system state and control input at time $k$ of the $j$-th iteration, and $T^j$ is the duration of the $j$-th iteration.

After $J$ number of iterations, we define the \emph{sampled safe state set} and \emph{sampled safe input set} as:
\begin{equation}\label{eq:SS}
\begin{aligned}
\mathcal{SS}^J = \textrm{}\left\{\bigcup_{j=1}^J {\bf{x}}^j \right\},~\mathcal{SU}^J = \textrm{}\left\{\bigcup_{j=1}^J {\bf{u}}^j \right\},
\end{aligned}
\end{equation}
where $\mathcal{SS}^J$ contains all states visited by the system in previous task iterations, and $\mathcal{SU}^J$ the corresponding inputs applied at each of these states.
Hence, by construction of the safe set, for any state in $\mathcal{SS}^J$ there exists a feasible input sequence contained in $\mathcal{SU}^J$ to reach the goal set $\mathcal{P}$ while satisfying state and input constraints (\ref{eq:VehicleModelConstr}). 

Similarly, we define the \emph{sampled cost set} as:
\begin{subequations} \label{eq:costSet}
    \begin{align}
    \mathcal{SQ}^J & = \textrm{}\left\{\bigcup_{j=1}^J {\bf{q}}^j \right\} \label{eq:sampledCost} \\ 
    {\bf{q}}^j &= [V^j(x_0^j), V^j(x_1^j), ..., V^j(x^j_{T^j})],
\end{align}
\end{subequations}
where $V^j(x_k^j)$ is the realized cost-to-go from state $x^j_k$ at time step $k$ of the $j$-th task execution:
\begin{equation}\label{eq:costToGo}
    V^j(x_k^j) = \sum_{i=k}^{T^j} h(x^j_i, u^j_i).
\end{equation}
The safe set based ILMPC policy tries to solve (\ref{eq:minTimeOC}) by using state and input data collected during past task iterations, stored in the sampled safe sets. At time $k$ of iteration $J+1$, we solve the optimal control problem:
\begin{align}\label{eq:minTimeLMPC}
    &V^{\mathrm{ILMPC}, J+1}(x_k^{J+1}) =  \\ 
    & \min_{u_{k|k}, ..., u_{k+N-1|k}} \sum_{t=k}^{k+N-1} h(x_{t|k}, u_{t|k}) + V^{J}(x_{t+N | k}) \nonumber \\
    & ~~~~~~~\mathrm{s.t.}~~~~~~~~~  x_{t+1 | k} = f(x_{t|k}, u_{t|k}), \forall t \in [k, k+N-1],\nonumber  \\
    & ~~~~~~~~~~~~~~~~~~~~ x_{t|k} \in \mathcal{X},~ u_{t|k} \in \mathcal{U},~~ \forall t \in [k, k+N-1], \nonumber \\
    & ~~~~~~~~~~~~~~~~~~~~ x_{k|k} = x_k^j,\nonumber \\ 
    & ~~~~~~~~~~~~~~~~~~~~ x_{k+N | k} \in \mathcal{SS}^{J} \cup \mathcal{P},\nonumber 
\end{align}
% \begin{subequations}\label{eq:ilmpc-policy}
%     \begin{align}
%      u = \pi^{\mathrm{ILMPC}}_{[1 \rightarrow M]}(x)& = \argmin_u ~ h(x,u) + q_k^j \\
%      & ~~~~~~~~\mathrm{s.t.}  ~u \in \mathcal{U} \\
%      & ~~~~~~~~~~~~ f(x,u) = x_k^j \in \mathcal{SS}^J_{[1 \rightarrow M]}\\
%      & ~~~~~~~~~~~~ q_k^j \in \mathcal{SQ}^J_{[1 \rightarrow M]},
%     \end{align}
% \end{subequations}
which searches for an input sequence over a chosen planning horizon $N$ that controls the system (\ref{eq:VehicleModel}) to the state in the sampled safe state set or task target set $\mathcal{P}$ with the lowest cost-to-go (\ref{eq:costToGo}). We then apply a receding horizon strategy:
\begin{align}\label{eq:ilmpc-policy}
    u(x^j_k) = \pi^{\mathrm{ILMPC}}(x^j_k)= u^\star_{k|k}.
\end{align}
A system (\ref{eq:VehicleModel}) in closed-loop with (\ref{eq:ilmpc-policy}) leads to a feasible task execution if $x_0 \in \mathcal{SS}^J$.
At each time step, the ILMPC policy searches for the optimal input based on previous task data, leading to performance improvement on the task as the sampled safe sets continue to grow with each subsequent iteration. For details on ILMPC, we refer to \cite{ugoprooftheorem}.

%%%%%% 
The sampled safe sets used in the ILMPC policy (\ref{eq:ilmpc-policy}) must first be initialized to contain at least one feasible task execution. 
%If the task changes even slightly, the trajectories stored in the safe sets will no longer be guaranteed to be feasible, so a new ILMPC policy must be re-initialized for the reconfigured task.
The aim of TDMPC is to use data collected from a task $\mathcal{T}1$ in order to efficiently find such an execution for a new task $\mathcal{T}2$. This will induce an ILMPC policy~(\ref{eq:ilmpc-policy}) that can be used to directly solve or initialize an ILMPC for $\mathcal{T}2$. 
We approach this using \emph{subtasks}, formalized in Sec.~\ref{ssec:formulation}, and the concept of controllability. 

\textit{Definition:} A system (\ref{eq:VehicleModel}) is \emph{N-step controllable} from an initial state $x_0$ to a terminal state $x_P$ if there exists an input sequence $[u_0, u_1, ..., u_{N-1}] \in \mathcal{U}$ such that the corresponding state trajectory satisfies state constraints (\ref{eq:VehicleModelConstr}) and $x_N = x_P$. A system is \emph{controllable} from $x_0$ to $x_P$ if there exists an $N>0$ such that the system is $N$-step controllable to $x_P$. \cite[Sec 10.1]{borrelliText}

\subsection{Subtasks}\label{ssec:formulation}
Consider an iterative task $\mathcal{T}$ (\ref{eq:task}) and a sequence of $M$ subtasks, where the $i$-th \textit{subtask} $\mathcal{S}_i$ is the tuple
\begin{align} \label{eq:subtaskdef}
    \mathcal{S}_i = \{\mathcal{X}_i, \mathcal{U}_i, \mathcal{R}_i\}.
\end{align}
We take $\mathcal{X}_i \subseteq \mathcal{X}$ as the subtask workspace, $\mathcal{U}_i \subseteq \mathcal{U}$ the subtask input space, and $\mathcal{R}_i$ the set of transition states from the current subtask $\mathcal{S}_i$ workspace into the subsequent $\mathcal{S}_{i+1}$ workspace:
\begin{equation}\label{eq:transitionset}
    \mathcal{R}_i \subseteq \mathcal{X}_i = \{x \in \mathcal{X}_i : \exists u \in \mathcal{U}_i, ~f(x, u) \in \mathcal{X}_{i+1}\}. \nonumber
\end{equation}
% \begin{equation}\label{eq:transitionset}
%     \mathcal{R}_i \subseteq \mathcal{X}_i \times \mathcal{U}_i = \{x \in \mathcal{X}_i, u \in \mathcal{U}_i : ~f(x, u) \in \mathcal{X}_{i+1}\}.
% \end{equation}
A successful \textit{subtask execution E($\mathcal{S}_i$)} of a subtask $\mathcal{S}_i$ is a trajectory of inputs and corresponding states evolving according to (\ref{eq:VehicleModel}) while respecting state and input constraints (\ref{eq:VehicleModelConstr}), ending in the transition set.
We define the $j$-th successful execution of subtask $\mathcal{S}_i$ as

\begin{subequations}\label{eq:subtaskexec}
    \begin{align}
E^j(\mathcal{S}_i) & = [{\bf{x}}_i^j, {\bf{u}}_i^j], \\
{\bf{x}}_i^j & = [x_{0}^j,x_{1}^j,~...,x_{T_i^j}^j], ~x_k \in \mathcal{X}_i ~~ \forall k \in [0, T^j_i], \nonumber\\
&~~~~~~~~~~~~~~~~~~~~~~~~~{x^j_{T_i^j}} \in \mathcal{R}_i, \label{eq:subtasktrans}\\
    {\bf{u}}_i^j & = [u_0^j,u_1^j,~...,u_{T_i^j}^j], ~u_k \in \mathcal{U}_i ~~ \forall k \in [0, T^j_i], \nonumber
\end{align}
\end{subequations}

% \begin{subequations}\label{eq:iterationvecs}
%     \begin{align}
% E^j(\mathcal{T}) & = [{\bf{x}}^j, {\bf{u}}^j] \\ 
%     {\bf{x}}^j & = [x_{0}^j,x_{1}^j,...,x^j_{T^j}], ~x^j_k \in \mathcal{X}~~ \forall k \in [0, {T^j}], \nonumber\\ & ~~~~~~~~~~~~~~~~~~~~~~~~~~~~~~~~~~ x_{T^j}^j \in \mathcal{P},\\
%     {\bf{u}}^j & = [u_0^j,u_1^j,...u^j_{T^j}], ~u^j_k \in \mathcal{U}~~ \forall k \in [0, {T^j}],
%     \end{align}
% \end{subequations}

where the vectors ${\bf{u}}_i^j$ and ${\bf{x}}_i^j$
collect the inputs applied to the system (\ref{eq:VehicleModel}) and the resulting states, respectively,  and $x_k^j$ and $u_k^j$ denote the system state and the control input at time $k$ of subtask execution $j$. $T_i^j$ is the duration of the $j$-th execution of subtask $i$.
The final state of each successful subtask execution is in the subtask transition set, from which it can evolve into new subtasks. For the sake of notational simplicity, we have written all subtask executions as beginning at time step $k=0$. %This need not be the case in practice. 

We say the task $\mathcal{T}$ is an ordered sequence of the $M$ subtasks ($\mathcal{T} = \{\mathcal{S}_i \}_{i=1}^M$) if the $j$-th successful task execution (\ref{eq:iterationvecs}) is equal to the concatenation of successful subtask executions:
\begin{subequations}
    \begin{align*}\label{eq:taskexecution}
    E^j(\mathcal{T}) & = [E^j(\mathcal{S}_1), E^j(\mathcal{S}_2), ..., E^j(\mathcal{S}_M)] = [{\bf{x}}^j, {\bf{u}}^j],\\
    {\bf{x}}^j & = [{\bf{x}}_1^j,{\bf{x}}_2^j,~...,{\bf{x}}_M^j],  \\
    {\bf{u}}^j & = [{\bf{u}}_1^j,{\bf{u}}_2^j,~...,{\bf{u}}_M^j], 
     \\
    f\Big(&x^j_{T^j_{[1\rightarrow i]}}, u_{T^j_{[1\rightarrow i]}} \Big)  \in \mathcal{X}_{i+1}, ~ i\in [1,M-1],\\
    x&^j_{T^j_{[1\rightarrow M]}} \in \mathcal{R}_M, 
    \end{align*}
\end{subequations}
where $T^j_{[1\rightarrow i]}$ is the duration of the first $i$ subtasks during the $j$-th task iteration.
When the state reaches a subtask transition set, the system has completed subtask $\mathcal{S}_i$, and it transitions into the following subtask $\mathcal{S}_{i+1}$. The task is completed when the system reaches the last subtask's transition set, $\mathcal{R}_M$, which we consider as the task's control invariant target set (referred to as $\mathcal{P}$ in the previous section).

\section{Task Decomposition for ILMPC}
Here we describe the intuition behind TDMPC, and provide an algorithm for the method. We also prove feasibility and iteration cost reduction of policies output by TDMPC. 

\subsection{TDMPC}
Let Task~$1$  and Task~$2$ be different ordered sequences of the same $M$ subtasks:
\begin{equation}
\begin{aligned}
    \mathcal{T}1 &=\{ \mathcal{S}_i \}_{i=1}^M, ~~ \mathcal{T}2 =\{ \mathcal{S}_{l_i} \}_{i=1}^M, \label{eq:twotasks}
\end{aligned}
\end{equation}
where the sequence $[l_1, l_2, ...,l_M]$ is a reordering of the sequence $[1,2,...,M]$.
Assume non-empty sampled safe sets $\mathcal{SS}_{[1\rightarrow M]}^J$, $\mathcal{SU}_{[1\rightarrow M]}^J$, and $\mathcal{SQ}_{[1\rightarrow M]}^J$ (\ref{eq:SS}, \ref{eq:costSet}) containing executions that solve $\mathcal{T}1$.

%an ILC in closed-loop with a system (\ref{eq:VehicleModel}) completes $J$ successful executions of Task~$1$, stored in the sampled safe sets $\mathcal{SS}_{[1\rightarrow M]}^J$ and $\mathcal{SU}_{[1\rightarrow M]}^J$. %Each stored execution contains a sequence of forward and backwards reachable states associated with the ILC policy at the $j$-th iteration.
The goal of TDMPC is to use the executions stored in the sampled safe sets from Task~$1$ in order to find a feasible execution for Task~$2$, ending in the new target set $\mathcal{R}_{l_M}$.
The key intuition of the method is that all successful subtask executions from Task~$1$ are also successful subtask executions for Task~$2$, as this definition only depends on properties (\ref{eq:subtaskexec}) of the subtask itself, not the subtask sequence.
Therefore, only stored points at subtask transitions need to be analyzed.
Based on this intuition, Algorithm \ref{alg:MBTL} proceeds backwards through the new subtask sequence $[l_1, l_2, ...,l_M]$. The key steps are discussed below.

\begin{algorithm}
\caption{TDMPC algorithm}\label{alg:MBTL}
\begin{spacing}{1.2}
\begin{algorithmic}[1]

\State \begin{varwidth}[t]{\linewidth} \textbf{input} $\mathcal{SS}_{[1 \rightarrow M]}^J, ~\mathcal{SU}_{[1 \rightarrow M]}^J, \mathcal{SQ}_{[1 \rightarrow M]}^J$, $[l_1, l_2, ..., l_M]$ \end{varwidth}
\State $\textbf{do} ~ \text{guard set clustering} (\mathcal{SS}_{[1 \rightarrow M]}^J) ~~~~~~~~~~~~~~~~~~~~~~(\ref{eq:guardsetcluster})$

\State \textbf{initialize empty }$\hat{\mathcal{SS}}, \hat{\mathcal{SU}}$

\State $ \hat{\mathcal{SS}}_{l_M} \gets~ \bigcup_{j=1}^J {\bf{{x}}}^j_{l_M}$ 

\State ${\bf{\hat{u}}}^j_{l_M} \gets {\bf{u}}^j_{l_M} \forall j \in [1,J], ~~ \hat{\mathcal{SU}}_{l_M} \gets~ \bigcup_{j=1}^J {\bf{\hat{u}}}^j_{l_M}$,

\State $\hat{\mathcal{SQ}}_{l_M} \gets~ \bigcup_{j=1}^J V^j({\bf{x}}^j_{l_M}) $

\For{$i \in [l_{M-1} : -1 : l_1]$} 
\State $\hat{\mathcal{SS}}_{i} \gets~ \bigcup_{j=1}^J {\bf{x}}^j_{i}$ 
\State ${\bf{\hat{u}}}^j_{i} \gets {\bf{u}}^j_{i}~  \forall j \in [1,J]$

\For{$x \in \mathcal{SG}_i$}
\State $k = \{k : x \in {\bf{x}}_i^k\}$  
\State \textbf{initialize empty } $Q^\star$
\For {$j : {\bf{{x}}}^j_{i+1} \in \hat{\mathcal{SS}}_{i+1}$}

\State ${\bf{\hat{q}}}^j_{i+1} = V^j( {\bf{x}}^j_{i+1})$

\State $\textbf{solve } (Q_j^\star, u_j^\star) = \mathrm{Ctrb}(x, {\bf{x}}^j_{i+1}, {\bf{\hat{q}}}^j_{i+1})$ ~(\ref{eq:feasibility})
\EndFor

\If{$Q^\star$ not empty}
% pick the best Q
\State $j_* = \argmin_j Q^\star_j$

\State ${\bf{\hat{u}}}^k_i[-1] \gets u_{j_*}^\star $

\Else
\State $\hat{\mathcal{SS}}_{i} \gets \hat{\mathcal{SS}}_{i} \text{\textbackslash} {\bf{x}}^k_{i} , ~~ \hat{\mathcal{SU}}_{i} \gets \hat{\mathcal{SU}}_{i} \text{\textbackslash} {\bf{\hat{u}}}^k_{i}$
%\State $\hat{\mathcal{SU}}_{i} \gets \hat{\mathcal{SU}}_{i} \text{\textbackslash} {\bf{\hat{u}}}^k_{i}$
\EndIf
\EndFor
\State $\hat{\mathcal{SU}}_{i} \gets~ \bigcup_{j=1}^J {\bf{\hat{u}}}^j_{i}, ~~ \hat{\mathcal{SQ}}_{i} \gets~ \bigcup_{j=1}^J V^j({\bf{\hat{x}}}^j_{i})$
%\State $\hat{\mathcal{SQ}}_{i} \gets~ \bigcup_{j=1}^J \mathcal{Q}({\bf{\hat{x}}}^j_{i}, {\bf{\hat{u}}}^j_{i})$
\EndFor

\State \begin{varwidth}[t]{\linewidth} \textbf{Return}
$\mathcal{SS}_{[l_1 \rightarrow l_M]}^0 = \bigcup_{i=l_1}^{l_M} \hat{\mathcal{SS}}_i$ \par 
\hskip \algorithmicindent \hskip \algorithmicindent $\mathcal{SU}_{[l_1 \rightarrow l_M]}^0 = \bigcup_{i=l_1}^{l_M} \hat{\mathcal{SU}}_i$ \par 
\hskip \algorithmicindent \hskip \algorithmicindent $\mathcal{SQ}_{[l_1 \rightarrow l_M]}^0 = \bigcup_{i=l_1}^{l_M} \hat{\mathcal{SQ}}_i$ 
\end{varwidth}

\end{algorithmic}
\end{spacing}
\end{algorithm}

\begin{itemize}
\item Consider subtask $\mathcal{S}_{l_M}$. 
All states from $\mathcal{S}_{l_M}$ stored in the Task~$1$ executions are controllable to $\mathcal{R}_{l_M}$ using stored inputs, i.e. there exists a stored input sequence that can be applied to the state such that the system evolves to be in $\mathcal{R}_{l_M}$. We next look for stored states from the preceding subtask,
$\mathcal{S}_{l_{M-1}}$, which are controllable to $\mathcal{R}_{l_M}$ via $\mathcal{S}_{l_M}$. (\underline{Algorithm \ref{alg:MBTL}, Lines 4-6})

Define the sampled guard set of $\mathcal{S}_{l_{M-1}}$ as
\begin{equation}\label{eq:guardsetcluster}
    \mathcal{SG}_{l_{M-1}} = \left\{ \bigcup_{j=1}^J x^j_{T^j_{[1 \rightarrow l_{M-1}]}} \right\}.
\end{equation}
The set contains those states in $\mathcal{S}_{l_{M-1}}$ from which the system transitioned into another subtask during one of the previous $J$ executions of $\mathcal{T}1$. Only controllability from the sampled guard set will be important in our approach. 

\item 
We search for the set of points in $\mathcal{SG}_{l_{M-1}}$ that are controllable to stored states in $\mathcal{S}_{l_M}$. This problem can be solved using a variety of numerical approaches. (\underline{Algorithm \ref{alg:MBTL}, Lines 9-15})

\item 
For any stored state $x$ in $\mathcal{SG}_{l_{M-1}}$ for which the controllability analysis failed, we remove the stored $\mathcal{SS}^J_{[l_{M-1}]}$ subtask execution ending in $x$ as candidate controllable states for Task~$2$. All remaining stored states in $\mathcal{S}_{l_{M-1}}$ are controllable to stored states in $\mathcal{S}_{l_M}$, and therefore also to $\mathcal{R}_{l_M}$. 
(\underline{Algorithm \ref{alg:MBTL}, Lines 15-20})
\end{itemize}

Algorithm ~\ref{alg:MBTL} iterates backwards through the remaining subtask sequence, connecting points in sampled guard sets to previously verified trajectories in the next subtask. Fig.~\ref{fig:algorithmpic} depicts this process across three subtasks from an autonomous racing task detailed in Section~\ref{ssec:trackexample}.
The algorithm terminates when it has iterated through the new subtask order, or when no states in a subtask's sampled guard set can be shown to be controllable to $\mathcal{R}_{l_M}$.
The algorithm returns sampled safe sets for Task~$2$ that have been verified through controllability to contain feasible executions of Task~$2$.
\begin{figure}
    \centering
    \includegraphics[width=0.3\textwidth]{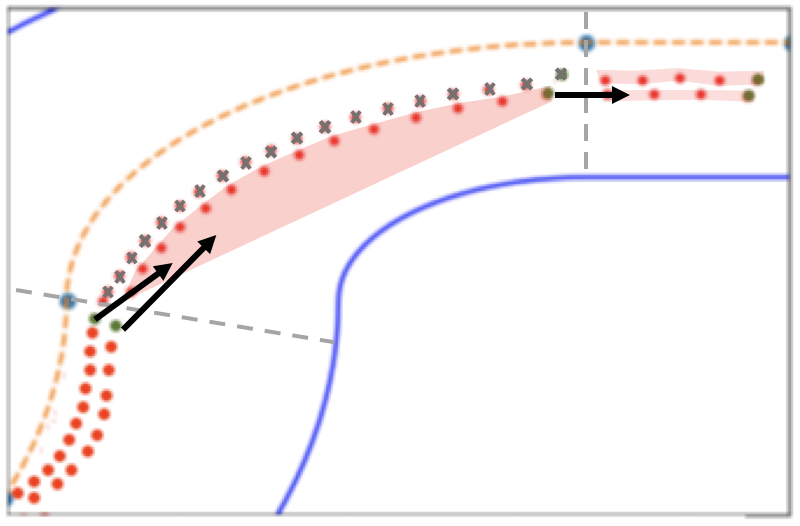}
    \caption{Algorithm ~\ref{alg:MBTL} checks controllability from states in the sampled guard set (green) to the convex hulls of safe trajectories through the next subtask (light red). If the controllability fails for a point in the sampled guard set, the backwards reachable states are removed from the safe set (grey). The centerlane is plotted in dashed yellow.}
    \label{fig:algorithmpic}
\end{figure}

TDMPC can improve on the computational complexity of existing trajectory transfer methods in two key ways: \textit{(i)}  by verifying stored trajectories only at states in the sampled guard set, rather than at each recorded time step, and \textit{(ii)} by solving a data-driven, one-step controllability problem to adapt the trajectories, rather than a multi-step or set-based controllability method.

\subsection{Properties of TDMPC-Derived Policies}
We prove feasibility and iteration cost reduction of ILMPC policies (\ref{eq:ilmpc-policy}) initialized using TDMPC. 
% Let an ILC, ``ILC $1$", be provided with an initial policy $u_k = \pi^c(x_k)$ for Task~$1$, and then complete $J$ executions of Task~$1$ in closed-loop with a linear, constrained system
% \begin{equation} \label{eq:linearsystem}
% \begin{aligned}
%     x_{k+1} &= Ax_k + Bu_k \\
%     \text{s.t.} & ~ x_k \in \mathcal{X } \\ 
%     & ~ u_k \in \mathcal{U },
% \end{aligned}
% \end{equation}
% where $\mathcal{X}$ and $\mathcal{U}$ are the convex state and input constraint sets, respectively.
% We form the safe sets $\mathcal{SS}_{[1 \rightarrow M]}^J$, $\mathcal{SU}_{[1 \rightarrow M]}^J$ out of the $J$ executions..
% \textit{Assumption 1}: When applied to $\mathcal{SS}_{[1 \rightarrow M]}^J$ and $\mathcal{SU}_{[1 \rightarrow M]}^J$, Algorithm \ref{alg:MBTL} returns nonempty 
% initial safe sets $\mathcal{SS}_{[l_1 \rightarrow l_M]}^0$, $\mathcal{SU}_{[l_1 \rightarrow l_M]}^0$.

% \begin{assumption}\label{ass:linear}
% The dynamics of a system (\ref{eq:VehicleModel}) are $f(x_k,u_k)$, and the system's state and input constraint sets (\ref{eq:VehicleModelConstr}) are $\mathcal{X}$ and $\mathcal{U}$, are convex. 
% \end{assumption}

\begin{assumption}\label{ass:convextasks}
Task~$1$ and Task~$2$ are defined as in (\ref{eq:twotasks}), where the subtask workspaces and input spaces are given by $\mathcal{X}_i = \mathcal{X},~\mathcal{U}_i = \mathcal{U}$ for all $i \in [1,M]$.
\end{assumption}

\begin{theorem}
	\label{th1}
	\textit{(Feasibility)} Let Assumptions
	~\ref{ass:controlInvariant}-\ref{ass:convextasks} hold. 
   Assume non-empty sets $\mathcal{SS}^J_{[1\rightarrow M]}$, $\mathcal{SU}^J_{[1\rightarrow M]}$, $\mathcal{SQ}^J_{[1\rightarrow M]}$ containing trajectories of system (\ref{eq:VehicleModel}) for Task~$1$. 
    Assume Algorithm ~\ref{alg:MBTL} outputs non-empty sets $\mathcal{SS}^0_{[l_1\rightarrow l_M]}$, $\mathcal{SU}^0_{[l_1\rightarrow l_M]}$, $\mathcal{SQ}^0_{[l_1\rightarrow l_M]}$ for Task~$2$. 
    Then, if $x_0 \in \mathcal{SS}^0_{[l_1\rightarrow l_M]}$, the policy $\pi^{\mathrm{ILMPC}}_{[l_1 \rightarrow l_M]}$, as defined in (\ref{eq:ilmpc-policy}), produces a feasible execution of Task~$2$. 
\end{theorem}

\begin{proof}
At every state $x_k$, the ILMPC policy (\ref{eq:ilmpc-policy}) searches for a sequence of inputs $[u_k, u_{k+1}, ..., u_{k+N-1}]$ such that, when applied to the system (\ref{eq:VehicleModel}), the resulting state $x_{k+N}$ is in $\mathcal{SS}^0_{[l_1 \rightarrow l_M]}$ or the target set $\mathcal{R}_{l_M}$. 

Since all states in $\mathcal{SS}^0_{[l_1 \rightarrow l_M]}$ are either stored as part of feasible trajectories to $\mathcal{R}_{l_M}$ or are directly in $\mathcal{R}_{l_M}$, such a sequence of inputs can always be found, and (\ref{eq:minTimeLMPC}) always has a solution:
\begin{equation}
\begin{aligned}
\forall~  &  x_k \in ~\mathcal{SS}_{[l_1 \rightarrow l_M]}^0, ~\exists~ [u_k, u_{k+1}, ..., u_{k+N-1}] \in ~\mathcal{U}: \nonumber \\ 
& x_{k+N} \in \mathcal{SS}_{[l_1 \rightarrow l_M]}^0  \subseteq \mathcal{X}. \nonumber
\end{aligned}
\end{equation}
As the terminal constraint set in \ref{eq:minTimeLMPC} is itself an invariant set, recursive feasibility follows from standard MPC arguments \cite{borrelliText}. It follows that the policy $\pi^{ILMPC}_{[l_1 \rightarrow l_M]}$ produces feasible trajectories for Task~$2$.
%If $x_{k+1} \in \mathcal{R}_{l_M}$, the task is complete via a feasible trajectory.
%Otherwise, $x_{k+1} \in \mathcal{SS}_{[l_1 \rightarrow l_M]}^0$. 
%For any such $x_{k+1}$, the ILMPC policy (\ref{eq:ilmpc-policy}) searches for a sequence of inputs such that, when applied to the system starting at $x_{k+1}$, the resulting state $x_{k+1+N}$ is in $\mathcal{SS}^0_{[l_1 \rightarrow l_M]}$. one such feasible sequence of inputs 

% It follows directly that
% \begin{align}
%     f(x_{k+1}, u_{k+1}) &\in \mathcal{SS}_{[l_1 \rightarrow l_M]}^0.
% \end{align}

%All states stored in the sampled safe sets are necessarily within the state constraint sets, and thus the policy $\pi^{\mathrm{ILMPC}}_{[l_1 \rightarrow l_M]}$ is feasible for all $k \geq 0$.
\end{proof}

The above Theorem \ref{th1} implies that the safe sets designed by the TDMPC algorithm induce an ILMPC policy that can be used to successfully complete Task~$2$ while satisfying all input and state constraints.

%%%%%%%%%%%%%%%%%%%%
% \textit{Assumption 2}: Assume the initial policy $u_k = \pi^c(x_k)$ is a subtask ordering invariant map from states $x$ to inputs $u$ that leads to successful task executions for any subtask ordering in closed-loop with the system (\ref{eq:VehicleModel}) (e.g. a centerline-tracking controller in the autonomous racing task). 
% \begin{theorem}
%     \label{th2}
%     \textit{(Iteration Cost Reduction)}
%     Let Assumption $2$ hold. Then, Algorithm ~\ref{alg:MBTL} will return non-empty sampled safe sets ${\mathcal{SS}}_{[l_1 \rightarrow l_M]}^{0}$, ${\mathcal{SU}}_{[l_1 \rightarrow l_M]}^{0}$ for Task~$2$. Furthermore, an ILC initialized with the induced policy $u_k = \pi_{T2ss}x(k)$ will incur no higher iteration cost than an ILC initialized with $u_k = \pi^c(x_k)$ during a task execution of Task~$2$.
% \end{theorem}
%%%%%%%%%%%%%%%%%%%%

\begin{assumption}
\label{ass:pidpolicy}
Consider Task~$1$ and Task~$2$ as defined in  (\ref{eq:twotasks}).  
The trajectories stored in $\mathcal{SS}^J_{[1\rightarrow M]}$ and  $\mathcal{SU}^J_{[1\rightarrow M]}$ correspond to executions of Task~$1$ by a nonlinear system (\ref{eq:VehicleModel}). One stored trajectory corresponds to an execution of (\ref{eq:VehicleModel}) in closed-loop with a policy $\pi^0(\cdot)$ that is feasible for both Task~$1$ and Task~$2$.
\end{assumption}

% An example of a control policy $\pi^0(\cdot)$ feasible for two different tasks for the autonomous racing task is a low speed  center-lane following controller. 
% %For the robotic manipulation task, a policy which always controls the end-effector to the center-height of the mode is feasible for any subtask order.
% If Assumption \ref{ass:pidpolicy} does not hold for a particular set of tasks $\mathcal{T}1$ and $\mathcal{T}2$, it is possible that Algorithm \ref{alg:MBTL} fails to find a feasible connection between trajectories of two subtasks. 
% In this case, the TDMPC strategy can be adapted as follows:
% \begin{enumerate}
%     \item The controllability problem (\ref{eq:feasibility}) may be transformed from a one-step to an $N$-step problem, where $N$ can be adjusted until a connection is found. 
%     \item A planning from scratch method can be used to search for a $\mathcal{T}2$ policy up until the failed subtask connection, after which an ILMPC-based policy (\ref{eq:ilmpc-policy}) based on the partial TDMPC safe set can be used for the remainder of the task.
% \end{enumerate}

\begin{theorem}
\label{th2}
\textit{(Cost Improvement)} 
    Let Assumptions ~\ref{ass:convextasks}-\ref{ass:pidpolicy} hold. 
    Then, Algorithm ~\ref{alg:MBTL} will return non-empty sets $\mathcal{SS}^0_{[l_1\rightarrow l_M]}$, $\mathcal{SU}^0_{[l_1\rightarrow l_M]}$, $\mathcal{SQ}^0_{[l_1\rightarrow l_M]}$ for Task~$2$.
    Furthermore, if $x_0 \in \mathcal{SS}^1_{[1 \rightarrow M]} $, an ILMPC initialized using $\mathcal{SS}^0_{[l_1 \rightarrow l_M]}$ will incur no higher iteration cost during an execution of Task~$2$ than an ILMPC initialized using a trajectory corresponding to (\ref{eq:VehicleModel}) in closed-loop with $\pi^0(\cdot)$.
\end{theorem}

\begin{proof}
Define the vectors
\begin{subequations}
\begin{align}
 {\bf{x}}^1 & = \mathcal{SS}_{[1 \rightarrow M]}^1 \subseteq \mathcal{SS}_{[1 \rightarrow M]}^J, \\
 {\bf{u}}^1 & = \pi^0({\bf{x}}^1) = \mathcal{SU}_{[1 \rightarrow M]}^1 \subseteq ~\mathcal{SU}_{[1 \rightarrow M]}^J,
\end{align}
\end{subequations}
to be the stored state and input trajectory associated with the implemented policy $\pi^0(\cdot)$.
Since $\pi^0(\cdot)$ is also feasible for Task~$2$, when Algorithm ~\ref{alg:MBTL} is applied, the entire task execution can be stored as a successful execution for Task~$2$ without adapting the policy. 
It follows that $\mathcal{SS}_{[1 \rightarrow M]}^1 \subseteq ~\mathcal{SS}^0_{[l_1 \rightarrow l_M]}$ and $ \mathcal{SU}_{[1 \rightarrow M]}^1 \subseteq \mathcal{SU}_{[l_1 \rightarrow l_M]}^0$, and the returned sample safe sets for Task~$2$ are non-empty.
% In accordance with \ref{eq:taskcost}, the cost incurred by a Task~$2$ execution is
% \begin{equation}
%     Q(E^1(\mathcal{T}2)) = h(x_0, u_0) + \mathcal{Q}(f(x_0,u_0), u_1).
% \end{equation}

At the initial state $x_0$, the ILMPC policy (\ref{eq:ilmpc-policy}) optimizes the chosen input so as to minimize the remaining cost-to-go. Consider an MPC planning horizon of $N=1$ (though this extends directly for any $N \geq 1$). 
Trivially,
\begin{align}
    &\min_u ~ h(x_{0},u) + V^j(x_p^j) ~~~~ \leq ~~~ h(x_{k},\pi^0(x_0)) + V^j(x_p^j)\nonumber \\
    & ~~\text{s.t.}   ~u \in \mathcal{U}, ~~~~~~~~~~~~~~~~~~~~~~~~~~ \text{s.t.} ~ f(x_{0},\pi^0(x_0)) = x_{p}^j.  \nonumber \\
    & ~~~~~ f(x_{0},u) = x_{p}^j. \nonumber
\end{align}
It follows that the cost incurred by a Task $2$ execution with $\pi^{ILMPC}_{[l_1 \rightarrow l_M]}$ is no higher than an execution with $\pi^{ILMPC}_{\pi^0}$. 
\end{proof}
The proof that the improved closed-loop iteration cost follows from the improved ILMPC cost (\ref{eq:minTimeLMPC}) is not presented here. 
However, the result shown holds for the examples presented in Sec.~\ref{ssec:trackexample}- \ref{ssec:rppexample}.

\subsection{Discussion}
In the examples presented in this paper, we implement the search for controllable points (Algorithm \ref{alg:MBTL}, Line 15) by solving a one-step controllability problem, $(Q^\star, u^\star)  = \mathrm{Ctrb}(x, \bf{z}, \bf{q})$, where
    \begin{align}
    u^\star, \lambda^\star &=\argmin_{u, \lambda} ~ h(x,u) + \lambda^\top\bf{q} \label{eq:feasibility}\\
&~~~~\textrm{s.t. }~f(x, u) = \lambda^\top\bf{z}, \nonumber\\
& ~~~~~~~~~\sum \lambda_i = 1,~ \lambda_i \geq 0 , \nonumber\\
& ~~~~~~~~~u \in \mathcal{U}_i,\nonumber \\ Q^\star &= \lambda^{\star\top} \bf{q}, \label{eq:convexCost}
    \end{align}
where \textbf{z} is a previously verified state trajectory through the next Task~$2$ subtask, and \textbf{q} the sampled cost vector associated with the trajectory.
(\ref{eq:feasibility}) aims to find an input that connects the sampled guard state $x$ to a state in the convex hull of the trajectory \cite[Sec 4.4.2]{borrelliText}.
If such an input is found, the new cost-to-go (\ref{eq:costToGo}) for the state $x$ is taken to be the convex combination of the stored cost vector (\ref{eq:convexCost}).
Solving the controllability analysis to the convex hull is an additional method for reducing computational complexity of TDMPC and is exact only for linear systems with convex constraints.

The points of subtask transition should be defined as is most useful, given the two tasks. 
Subtask transition points simply indicate which segments of the stored trajectories are certain to remain feasible in $\mathcal{T}2$ using the stored policies - but this can change depending on how exactly $\mathcal{T}2$ differs from $\mathcal{T}1$.
The TDMPC method is therefore not limited in applicability to a predetermined number of reshuffled tasks. 

\section{Application 1: Autonomous Racing}\label{ssec:trackexample}

\begin{figure}[t!]
    \centering
    \includegraphics[scale=0.45]{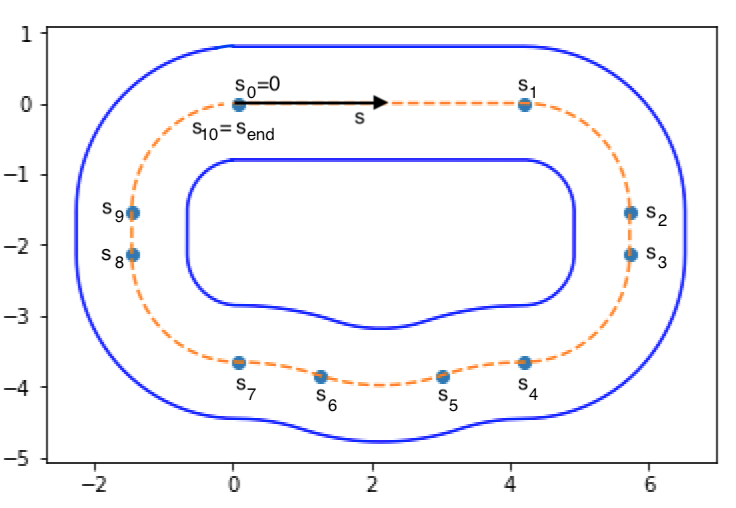}
    \caption{Each subtask of the racing task corresponds to a segment of the track with constant curvature. The vehicle state $s$ tracks the distance traveled along the centerline.}
    \label{fig:tracktask}
\end{figure}
\subsection{Task Formulation}
Consider an autonomous racing task, in which a vehicle is controlled to minimize lap time driving around a race track with piecewise constant curvature (Fig.~\ref{fig:tracktask}). We model this task as a series of ten subtasks, where the $i$-th subtask corresponds to a section of the track with constant radius of curvature $c_i$. Tasks with different subtask order are tracks consisting of the same road segments in a different order.

The vehicle is modeled in the curvilinear abscissa reference frame \cite{rajamani2011vehicle}, with states and inputs at time step $k$
\begin{subequations}
\begin{align}
    x_k & = [v_{x_k} ~ v_{y_k} ~ \dot\psi_{k} ~ e_{\psi_k} ~ s_k ~ e_{y_k} ]^\top,\nonumber\\
    u_k &= [a_{k} ~ \delta_k ]^\top, \nonumber
\end{align}
\end{subequations}
where $v_{x_k}$, $v_{y_k}$, and $\dot{\psi}_{k}$ are the vehicle's longitudinal velocity, lateral velocity, and yaw rate, respectively, at time step $k$, $s_k$ is the distance travelled along the centerline of the road, and $e_{\psi_k}$ and $e_{y_k}$ are the heading angle and lateral distance error between the vehicle and the path. The inputs are longitudinal acceleration $a_k$ and steering angle $\delta_k$. The system dynamics (\ref{eq:VehicleModel}) are described using an Euler discretized dynamic bicycle model \cite{rosolia2017autonomousrace}.
Accordingly, the system state and input spaces are 
\begin{subequations}
    \begin{align}
    \mathcal{X} &= \mathbb{R}^6, ~ \mathcal{U} = \mathbb{R}^2. \nonumber
    \end{align}
\end{subequations}
We formulate each subtask according to (\ref{eq:subtaskdef}), with:
\subsubsection{Subtask Workspace $\mathcal{X}_i$}
\begin{equation}
    \mathcal{X}_i = \left\{ x : \begin{bmatrix}0 \\ -\frac{\pi}{2}~ \mathrm{rad} \\ s_{i-1} \\ -\frac{l}{2} \end{bmatrix} \leq  \begin{bmatrix}v_x \\ e_{\psi} \\ s \\ e_y \end{bmatrix}  \leq   \begin{bmatrix}3~\mathrm{m/s} \\ \frac{\pi}{2} ~\mathrm{rad}\\ s_{i} \\ \frac{l}{2} \end{bmatrix} \right \}, \nonumber
\end{equation}
where $s_{i-1}$ and $s_{i}$ mark the distances along the centerline to the start and end of the curve, and $l = 0.8 $ is the lane width in meters. $s_{10} = s_{\mathrm{end}}$ is the total length of the track.
These bounds indicate that the vehicle can only drive forwards on the track, up to a maximum velocity, and must stay within the lane. 

\subsubsection{Subtask Input Space $\mathcal{U}_i$}
\begin{equation}
    \mathcal{U}_i = \begin{bmatrix}-1~{\mathrm{m/s^2}} \\ -0.5~{\mathrm{rad/s^2}} \end{bmatrix} \leq  \begin{bmatrix}a \\ \delta \end{bmatrix}  \leq   \begin{bmatrix}1~{\mathrm{m/s^2}} \\ 0.5~{\mathrm{rad/s^2}} \end{bmatrix}. \nonumber
\end{equation}
The input limits are a function of the vehicle, and do not change between subtasks.

\subsubsection{Subtask Transition Set, $\mathcal{R}_i$}
Lastly, we define the subtask transition set to be the states along the subtask border where the track's radius of curvature changes:
\begin{equation}
    \mathcal{R}_i = \{x \in \mathcal{X}_i: \exists u\in\mathcal{U}_i \text{, s.t. } s^+ = s_{i+1}\}, \nonumber
\end{equation}
where $x^+ = f(x,u)$.
The task target set is the race track's finish line, 
\begin{equation}
    \mathcal{R}_M = \left\{ x : s \geq s_{\mathrm{end}} \right \}. \nonumber
\end{equation}
The task goal is to complete a lap and reach the target set as quickly as possible. Therefore we define the stage cost as:
\begin{equation}
    h(x_k, u_k) = \begin{cases}
			0, & x_k \in \mathcal{P}\\
            1, & \mathrm{otherwise}.
		 \end{cases} \nonumber
\end{equation}

\subsection{Simulation Setup}
An ILMPC (\ref{eq:ilmpc-policy}) is used to complete $J=5$ executions of Task~$1$, the track depicted in Fig.~\ref{fig:tracktask}. The vehicle begins each task iteration at standstill on the centerline at the start of the track. The $J$ executions and their costs are stored in $\mathcal{SS}_{[1 \rightarrow M]}$, $\mathcal{SU}_{[1 \rightarrow M]}$, and $\mathcal{SQ}_{[1 \rightarrow M]}$.
An initial trajectory for the ILMPC safe sets is executed using a centerline-tracking, low-velocity PID controller, $\Pi_0$. 
%The task parameters are chosen as $l = 1.6 ~\mathrm{m}$, $v_{\mathrm{max}} = 0.5~\mathrm{m/s}$, $ -1 ~\mathrm{m/s^2} \leq a \leq 1 ~\mathrm{m/s^2}$, and $ -0.5 ~\mathrm{rad} \leq \delta \leq 0.5 ~\mathrm{rad}$.

TDMPC then uses these sampled safe sets to design initial policies for a new track composed of the same track segments.
Two ILMPCs are designed for the reconfigured track: one initialized with TDMPC, and another initialized with $\Pi_0$. Each ILMPC completes $J=10$ laps around the new tracks.
In this examples, the reconfigured track is not continuous, and should be considered to be a segments of larger, continuous track.

\subsection{Simulation Results}
\begin{figure}
\centering
\begin{subfigure}[b]{0.5\textwidth}
        \centering
        \includegraphics[width=0.75\textwidth]{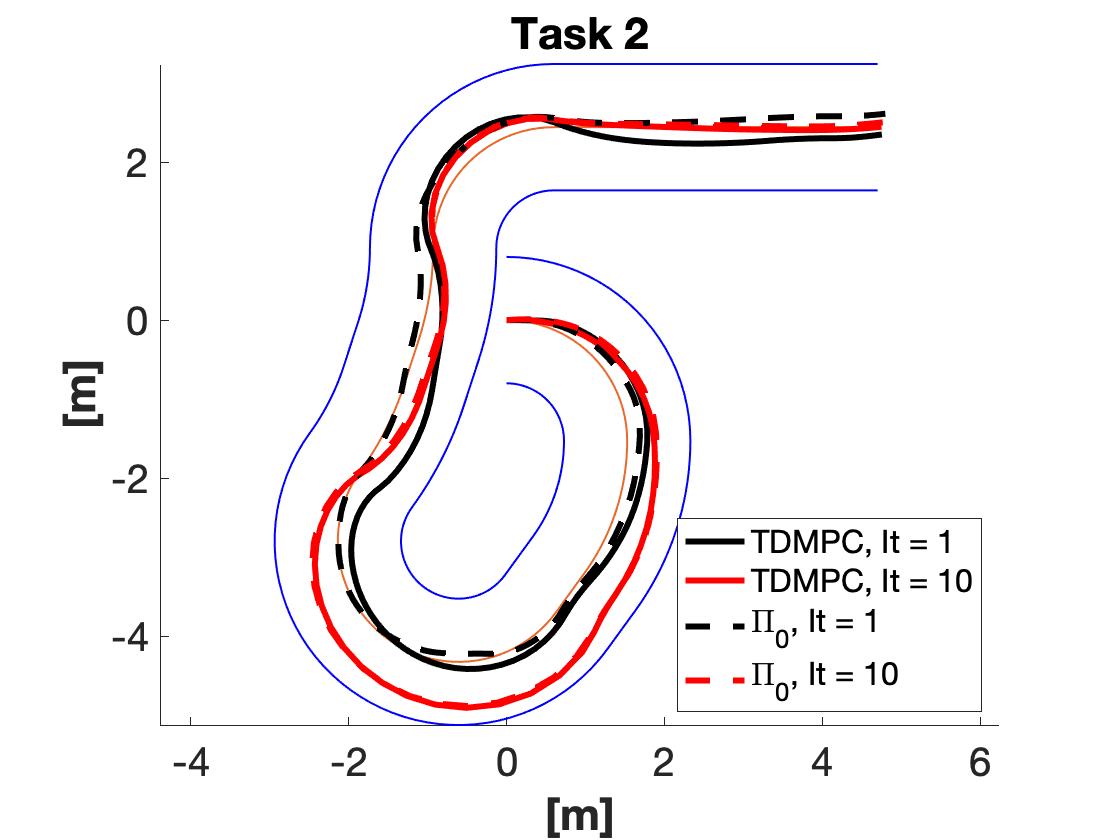}
\end{subfigure}
\begin{subfigure}[b]{0.5\textwidth}
        \centering
        \includegraphics[width=0.75\textwidth]{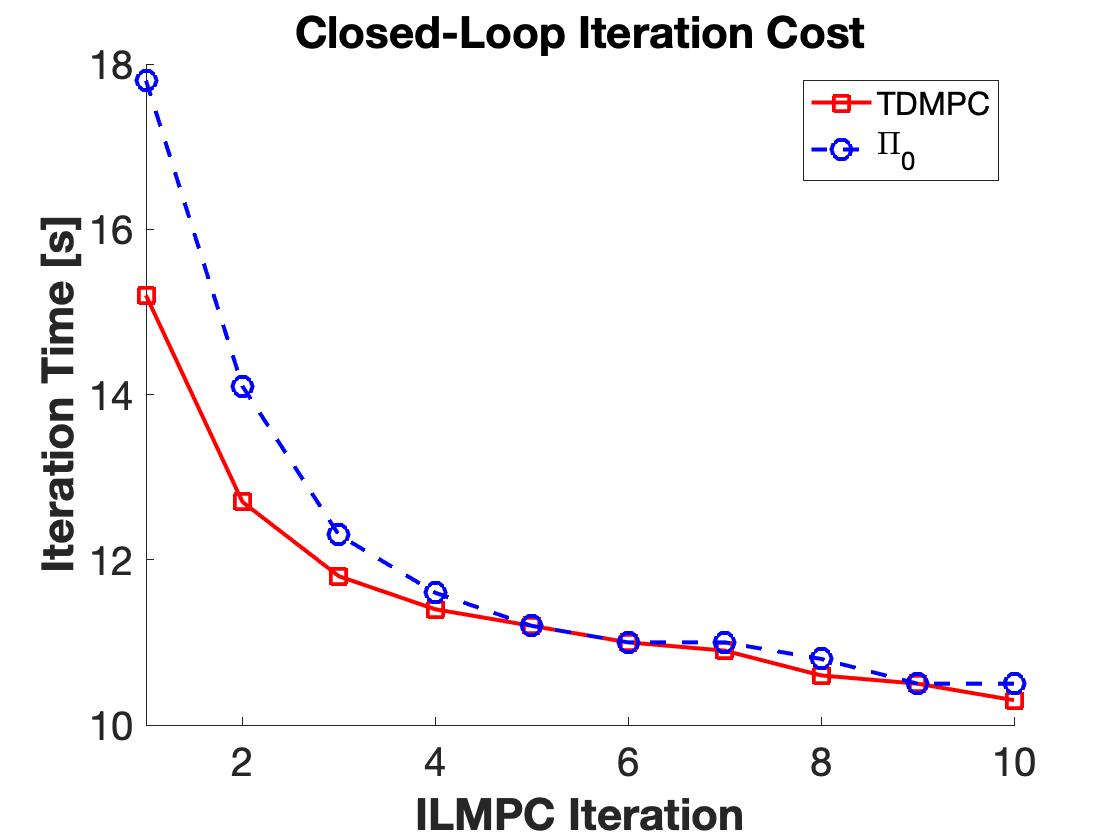}
\end{subfigure}
\caption{The TDMPC-initialized ILMPC converges to a locally optimal trajectory faster than the PID-initialized one.}
\label{fig:singleShuffleResults}
\end{figure}

Fig.~\ref{fig:singleShuffleResults} compares the first and tenth trajectories around the track of the two ILMPCs, plotted as black and red lines. 
The $\Pi_0$-initialized ILMPC (in dashed black) initially stays close to the centerline, taking nearly $18$ seconds to traverse the new track. The TDMPC-initialized ILMPC, however, traverses the new track more efficiently starting with the first lap. 
The first lap completed using the TDMPC-initialized ILMPC (in solid black) begins  closer to the final locally optimal policies (in red) that both ILMPCs eventually converge to. 
In this example, the TDMPC method is able to leverage experience on another track in order to complete sections of the new track in a locally optimal way, even on the first iteration of a new task. 

The lap time of each of the ten ILMPC iterations is plotted in the bottom of Fig.~\ref{fig:singleShuffleResults}.
As expected, the TDMPC-initialized ILMPC completes the first several laps faster than the $\Pi_0$-initialized ILMPC. 
The TDMPC-initialized ILMPC requires fewer task iterations and less time per iteration to reach a locally optimal trajectory. 

\section{Application 2: Robotic Path Planning}\label{ssec:rppexample}

TDMPC can also be used to combine knowledge gained from solving a variety of previous tasks.
For example, if $D$ ILMPCs as in (\ref{eq:ilmpc-policy}) complete $J$ iterations of $D$ different tasks, TDMPC can be used to design a policy for a task $(D+1)$. The algorithm draws on subtask executions collected over $D$ different tasks in order to build safe sets for Task $D+1$. 
We evaluated this approach in a robotic path planning example. 
\begin{figure}
    \centering
    \includegraphics[scale=0.3]{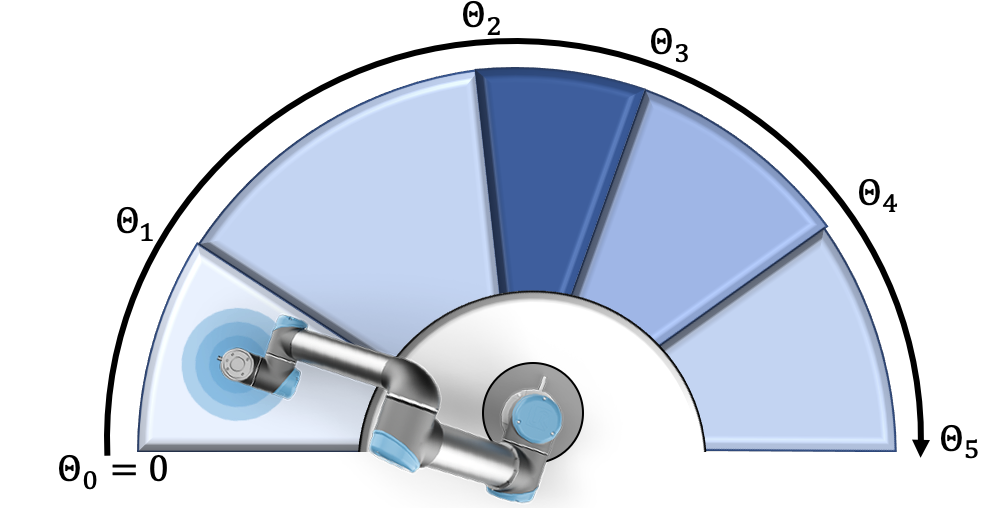}
    \caption{Topview of the robotic path planning task. Each subtask corresponds to an obstacle in the environment with constant height.}
    \label{fig:robottaskpic}
\end{figure}

\subsection{Task Formulation}
Consider a task in which a UR5e \footnote{https://www.universal-robots.com/products/ur5-robot/} robotic arm needs to move an object to a target without colliding with obstacles (Fig.~\ref{fig:robottaskpic}).
The obstacles are modeled as extruded disks of varying heights above and below the robot, leaving a workspace space between $h_{\mathrm{min},i}$ and $h_{\mathrm{max},i}$.
Here, each subtask corresponds to the workspace above a particular obstacle. 
Different subtask orderings correspond to a rearranging of the obstacle locations.

The end-effector reference tracking accuracy of the UR5e allows us to use a simplified model in robot experiments, in place of a discretized second-order model as in \cite{robothard3}. We solve the task in the reduced state space:
\begin{subequations}
\begin{align}
    {x}_k & = [q_{0_k} ~ \dot{q}_{0_k} ~ z_k ~ \dot{z}_k]^\top, \nonumber \\
    {u}_k &= [\ddot{q}_{0_k} ~ \ddot{z}_k]^\top, \nonumber
\end{align}
\end{subequations}
where $z_k$ is the height of the robot end-effector at time step $k$, calculated from the joint angles via forward kinematics, and $\dot{z}_k$ the upward velocity. We control $\ddot{q}_{0_k}$ and $\ddot{z}_k$, the accelerations of $q_0$ and $z$, respectively. 
The system state and input spaces are 
\begin{subequations}
    \begin{align}
    \mathcal{X} &= \mathbb{R}^{4}, ~  \mathcal{U} = \mathbb{Rb}^2. \nonumber
    \end{align}
\end{subequations}
We model the simplified system as a quadruple integrator:
\begin{equation}\label{eq:simplifiedRobotModel}
    {x}_{k+1} = \begin{bmatrix}1 & dt & 0 & 0 \\ 0 & 1 & 0 & 0\\ 0 & 0 & 1 & dt \\ 0 & 0 & 0 & 1 \end{bmatrix} {x}_k + \begin{bmatrix}0 & 0 \\ dt & 0 \\ 0 & 0 \\ 0 & dt \end{bmatrix} {u}_k,
\end{equation}
where $dt = 0.01$ seconds is the sampling time. This simplified model holds as long as we operate within the region of high end-effector reference tracking accuracy, characterized in previous experiments.
% \begin{subequations}
% \begin{align}
% \tilde{\mathcal{X}} &= \begin{bmatrix}a_{min} \\ \delta_{min} \end{bmatrix} \leq  \begin{bmatrix}\dot{q}_{0k} \\ \dot{z}_k  \end{bmatrix}  \leq   \begin{bmatrix}a_{max} \\ \delta_{max} \end{bmatrix}.
% \end{align}
% \end{subequations}
% The robotic manipulator is modeled as a six-joint robotic arm, with states and inputs
% \begin{subequations}
% \begin{align}
%     x_k & = [q_{0_k} ~ \dot{q}_{0_k} ~ q_{1_k} ~ \dot{q}_{1_k} ~ q_{2_k} ~ \dot{q}_{2_k} ~q_{3_k} ~ \dot{q}_{3_k} ~q_{4_k} ~ \dot{q}_{4_k} ~ q_{5_k} ~ \dot{q}_{5_k}]^\top\\ 
%     u_k &= [\tau_{0_k} ~ \tau_{1_k} ~ \tau_{2_k}~ \tau_{3_k}~ \tau_{4_k}~\tau_{5_k}]^\top,
% \end{align}
% \end{subequations}
% where $q_{i_k}$ and $\dot{q}_{i_k}$ are the angle and angular velocity of the $i$-th joint, respctively, at time step $k$. 
% The inputs are the torques $\tau_{i_k}$ applied at each of the joints. 

% The continus-time system dynamics are given by:
% \begin{equation} \label{eq:robotdynamics}
% \begin{aligned}
% M(x)\dot{x_k} + C(x)x_k + g(x) = u_k,
% \end{aligned}
% \end{equation}
% where $M$ is the mass inertia matrix, $C$ the matrix of Coriolis and centrifugal forces, and $g$ the vector of gravity terms. We refer to \cite{robothard3} for details and the discretized form of (\ref{eq:robotdynamics}).

We formulate each subtask according to (\ref{eq:subtaskdef}).

\subsubsection{Subtask Workspace $\mathcal{X}_i$}
\begin{equation}
\begin{aligned}
    \tilde{\mathcal{X}}_i =  \begin{bmatrix}\theta_{i-1} \\ -\pi ~ \mathrm{rad/s} \\ h_{\mathrm{min}, i} \\ -1 ~ \mathrm{m/s} \end{bmatrix} \leq  \begin{bmatrix}q_0 \\ \dot{q}_{0k} \\ z_k \\ \dot{z}_k  \end{bmatrix}  \leq   \begin{bmatrix}\theta_i \\ \pi ~ \mathrm{rad/s} \\ h_{\mathrm{max},i} \\ 1 ~
    \mathrm{m/s} \end{bmatrix} \nonumber
\end{aligned}
\end{equation}
where $\theta_{i-1}$ and $\theta_{i}$ mark the cumulative angle to the beginning and end of the $i$-th obstacle, as in Fig.~\ref{fig:robottaskpic}. States $\dot{q}_{0k}$ and $\dot{z}_k$ are constrained to lie in the experimentally determined region of high end-effector tracking accuracy. The robot end-effector is constrained not to collide with the subtask obstacle.

\subsubsection{Subtask Input Space $\mathcal{U}_i$}
\begin{equation}
 \tilde{\mathcal{U}}_i = \begin{bmatrix} -\pi ~ \mathrm{rad/s^2} \\ -0.6 ~ \mathrm{m/s^2} \end{bmatrix} \leq  \begin{bmatrix} \ddot{q}_{0k} \\ \ddot{z}_k  \end{bmatrix}  \leq   \begin{bmatrix}\pi ~ \mathrm{rad/s^2} \\ 0.6 ~ \mathrm{m/s^2} \end{bmatrix}, \nonumber
\end{equation}
where $\ddot{q}_{0k}$ and $\ddot{z}_k$ are constrained to lie in the experimentally determined region of high end-effector tracking accuracy.

\subsubsection{Subtask Transition Set, $\mathcal{R}_i$}
We define the subtask transition set to be the states along the subtask border where the next obstacle begins:
\begin{equation}
    \mathcal{R}_i = \{x \in \mathcal{X}_i: \exists u\in\mathcal{U}_i \text{, s.t. } q_0^+ \geq \theta_{i}~ \}, \nonumber
\end{equation}
where $x^+ = f(x,u)$.
The task target set is the end of the last mode:
\begin{equation}
    \mathcal{R}_M = \left\{ x : q_0 = \theta_M,~h_{\mathrm{min},M} \leq z \leq h_{\mathrm{max},M} \right \}.\nonumber
\end{equation}
The task goal is to reach the target set as quickly as possible:
\begin{equation}
    h(x_k, u_k) = \begin{cases}
			0, & x_k \in \mathcal{P}\\
            1, & \mathrm{otherwise}.
		 \end{cases}\nonumber
\end{equation}

% \begin{centering}
% \begin{figure*}[h!]
% \begin{tabular}{cc}
%   \includegraphics[width=57mm]{final_task1_obstacle.jpg} & \includegraphics[width=57mm]{final_task3_obstacle.jpg} \\
% \includegraphics[width=57mm]{final_task1_cost.jpg} & \includegraphics[width=57mm]{final_task3_cost.jpg} \\
% (a)  & (b)\\[3pt]
% \end{tabular}
% \caption{Caption}
% \label{fig:robotShuffleResults}
% \end{figure*}
% \end{centering}

\subsection{Experimental Setup}
An ILMPC (\ref{eq:ilmpc-policy}) was used to complete $J=10$ executions of five different training tasks, where each training task corresponded to a reordering of the obstacles. 
In each task, the ILMPC tries to reach the target set as quickly as possible while avoiding the obstacles.
Each ILMPC was initialized with a trajectory resulting from executing a policy $\Pi_0$ that tracks the center height of each mode with the end-effector, while the robot rotated at a low constant joint velocity $\dot{q}_0$. 
TDMPC was then applied to the combined sampled safe sets of the five training tasks, and used to design an initial policy for a new ILMPC on an unseen ordering of obstacles, shown in Fig.~\ref{fig:robotResults}. The white space corresponds to environment obstacles, so that the ILMPC task is to reach the end of the last mode as quickly as possible while controlling the end-effector to remain within the safe (green) part of the state space. A second ILMPC was initialized with the center-height tracking $\Pi_0$, for comparison. 
After initialization, the two ILCMPs completed $J=20$ iterations of the new task. These iterations were executed in simulation using the simplified model (\ref{eq:simplifiedRobotModel}), and the first and last trajectories of each ILMPC were then tracked by a real UR5e robot using end-effector tracking. 

\subsection{Experimental Results}
% \begin{figure}
% \centering
% \subfloat{ \includegraphics[width=0.35\textwidth]{accRobot.jpg} } 
% \subfloat{ \includegraphics[width=0.35\textwidth]{accRobotCost.jpg} }
% \caption{The TDMPC-initialized ILMPC solves $\mathcal{T}2$ much faster than the ILMPC initialized with a center-height tracking policy $\Pi_0$.}
% \label{fig:robotResults}
% \end{figure}

\begin{figure}
\centering
\begin{subfigure}[b]{0.5\textwidth}
        \centering
        \includegraphics[width=0.75\textwidth]{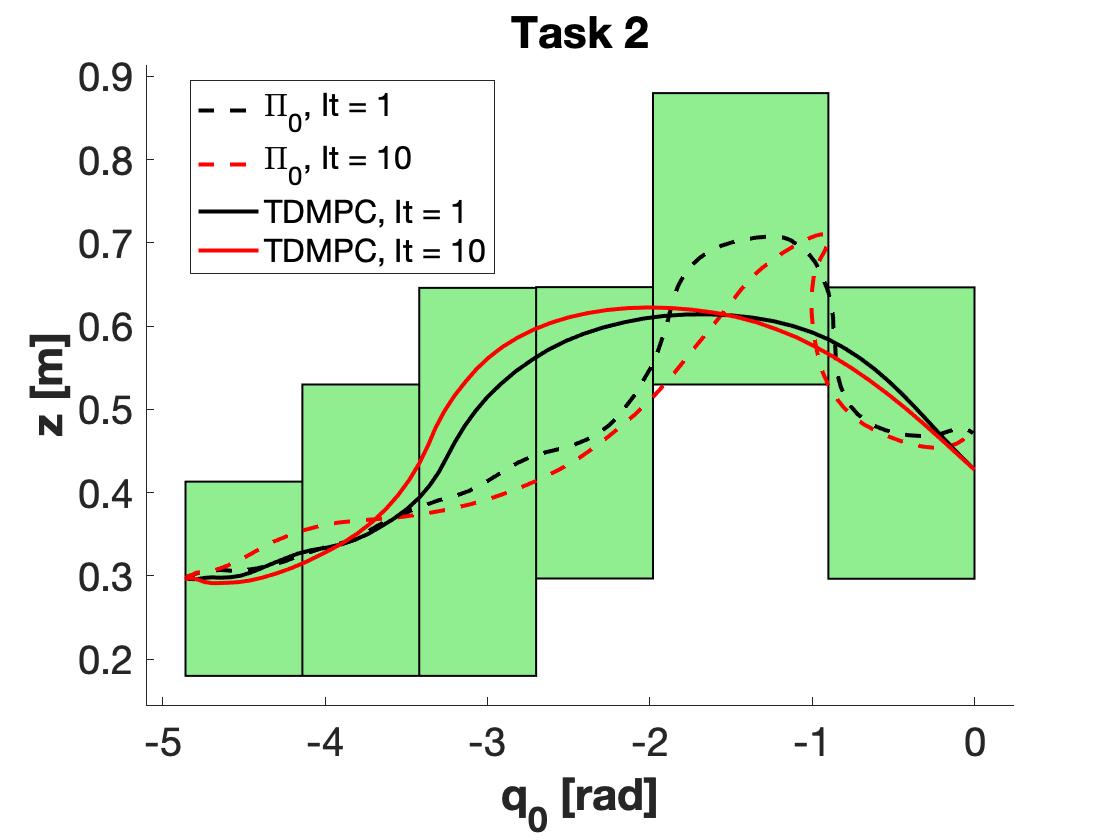}
\end{subfigure}
\begin{subfigure}[b]{0.5\textwidth}
        \centering
        \includegraphics[width=0.75\textwidth]{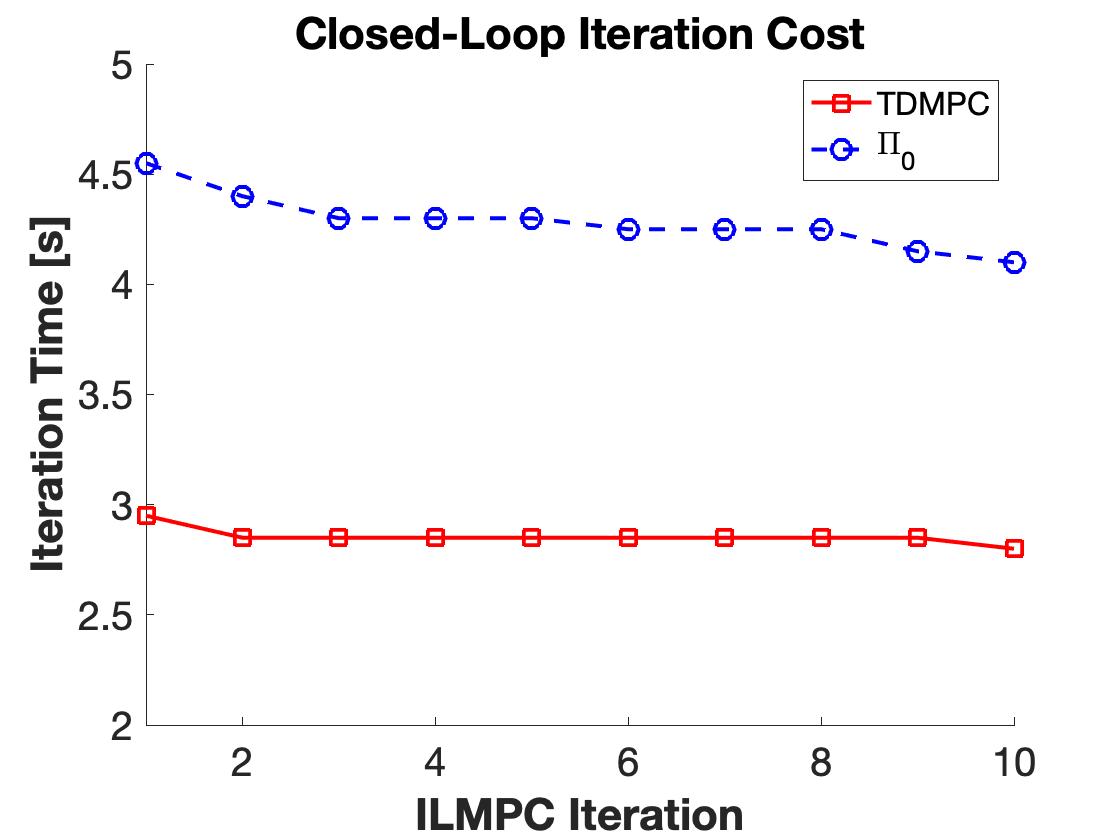}
\end{subfigure}
\caption{The TDMPC-initialized ILMPC solves $\mathcal{T}2$ much faster than the ILMPC initialized with a center-height tracking policy $\Pi_0$.}
\label{fig:robotResults}
\end{figure}

The measured robot trajectories are plotted in  Fig.~\ref{fig:robotResults}. 
The $\Pi_0$-initialized ILMPC follows the center-height of each mode closely during the first task iteration (plotted in dashed black). After ten iterations of the task, the resulting trajectory (plotted in dashed red) has only diverged from the center-height trajectory slightly. Correspondingly, after ten iterations the $\Pi_0$-initialized ILMPC still requires more than four seconds to complete the task.

The TDMPC-initialized ILMPC, however, draws on knowledge gathered over many previous tasks in order to solve the task efficiently right away. Already on the first trajectory (plotted in solid black), the TDMPC-initialized ILMPC solves the task in under three seconds. This is a $30\%$ improvement over the $\Pi_0$-initialized ILMPC.
As in the autonomous driving task, the first trajectory completed by the TDMPC-initialized ILMPC is very close to the ultimate locally optimal trajectory. 

Because of the nonconvex obstacles, this task is nonconvex, and there are many locally optimal trajectories. At various iterations of the task, both the TDMPC-initialized and the $\Pi_0$-initialized ILMPCs get stuck at such local minima, so that the ILMPCs performance metric remained constant over several iterations before improving again (Fig.~\ref{fig:robotResults}). At these performance plateaus, the realized trajectories continue to change. We believe that the variability in the mixed integer solver used in the ILMPC led the ILMPC to follow different trajectories with the same iteration cost, as if encouraging exploration. Some of these different trajectories then allowed for performance improvement in the next iteration.

%%%%%%%%%%%%%%%%%%%%%%%%%%%%%%%%%%%%%%%%%%%%%%%%%%%%%%%%%%%%%%%%%%%%%%%%%%%%%%%%
\section{Conclusion}
A task decomposition method for ILMPC was presented.
The TDMPC algorithm uses stored state and input trajectories from executions of a task, and efficiently designs policies for executing variations of that task. TDMPC breaks tasks into subtasks and performs controllability analysis at sampled safe states between subtasks. 
The algorithm can improve upon other methods by only needing to verify and adapt the original task policy at points of subtask transition, rather than along the entire trajectory.
We evaluate the effectiveness of the proposed algorithm on autonomous racing and robotic manipulation tasks. Our results confirm that TDMPC allows an ILMPC to converge to a locally-optimal minimum-time trajectory faster than using simple methods. 

% \section{Appendix}
% \begin{table}[h!]
% \caption{Lengths $l_i$ and piecewise constant radii of curvature $c_i$ of each track subtask segment $i$.}
% \label{tab:track1}
% \begin{center}
% \begin{tabular}{|c||c|c|c|c|c|}
% \hline
% \textbf{i} & \textbf{1} & \textbf{2} & \textbf{3} & \textbf{4} & \textbf{5} \\
% \hline
% \textbf{$l_i$} $m$ & $4.2$ & $2.4$ & $0.6$ & $2.4$ & $1.2$\\
% \hline
% \textbf{$c_i$} [m/$\pi$] & $0$ & $-4.8$ & $0$ & $-4.8$ & $12$\\
% \hline
% \hline
% \textbf{i} & \textbf{6} & \textbf{7} & \textbf{8} & \textbf{9} & \textbf{10} \\
% \hline
% \textbf{$l_i$} [m] & $1.8$ & $1.2$ & $2.4$ & $0.6$ & $2.4$\\
% \hline
% \textbf{$c_i$} [m/$\pi$] & $9$ & $12$ & $-4.8$ & $0$ & $-4.8$\\
% \hline
% \end{tabular}
% \end{center}
% \end{table}

%\renewcommand{\baselinestretch}{0.92}
\bibliographystyle{IEEEtran}
\bibliography{IEEEabrv,mybib}

\end{document}